\documentclass[a4paper]{article}
\usepackage[utf8]{inputenc}
\usepackage{amsmath,amssymb,amsthm,bbold,colortbl,enumerate,float,mathrsfs,mathtools,mdframed,physics,todonotes,url,xcolor}
\usepackage[top=2.5cm, bottom=2.5cm, left=2.4cm, right=2.4cm]{geometry}
\usepackage{algorithm,algpseudocode}
\usepackage[hidelinks]{hyperref}

\usepackage[noabbrev]{cleveref}
\newtheorem{theorem}{Theorem}
\newtheorem{lemma}[theorem]{Lemma}

\theoremstyle{definition}

\newtheorem{problem}{Open problem}
\newtheorem{definition}[theorem]{Definition}

\newcommand{\CC}{\mathbb C}

\newcommand{\id}{\mathbb{1}}

\title{Thirty-six quantum officers are entangled}
\author{Simeon Ball \\ \textit{Universitat Politècnica de Catalunya} \and Robin Simoens \\ \textit{Ghent University} \\ \textit{Universitat Politècnica de Catalunya}}
\date{}

\begin{document}

\maketitle

\begin{abstract}
    There exist pairs of orthogonal Latin squares of any order \(n\) except if \(n=2\) or \(n=6\) [Bose, Shrikhande and Parker, 1960]. In particular, the problem of Euler's thirty-six officers does not have a solution. However, it has a ``quantum solution'': there exist so-called entangled quantum Latin squares of order six [Rather et al., 2022]. We prove that mutually orthogonal quantum Latin squares of order six do not exist if entanglement is not allowed.
\end{abstract}



\section{Introduction}

Musto and Vicary \cite{MustoVicary2016} introduced the following quantum version of Latin squares.

\begin{definition}[\cite{MustoVicary2016}]\label{def:QLS}
    A \emph{quantum Latin square of order \(n\)} is an \(n\times n\) matrix over \(\CC^n\) such that each row and each column forms an orthonormal basis of \(\CC^n\).
\end{definition}

In the literature, there are at least four different ways to define the quantum analogue of orthogonal Latin squares.
The first such definition was given by Musto \cite[Definition~10]{Mustopaper}, calling them \emph{weak orthogonal} or \emph{left orthogonal} quantum Latin squares. Goyeneche, Raissi, Di Martino and \.Zyczkowski \cite{Goyeneche} called the notion \emph{weakly orthogonal}, and pointed out that it does not seem like the right definition, proposing an alternative \cite[Definition~3]{Goyeneche}.
Later, Musto \cite[Definition~5.2.1]{Mustothesis} gave another concept of orthogonal quantum Latin squares and proved it to be equivalent to the one in \cite{Goyeneche} if the quantum Latin squares are \emph{non-entangled}. To us, this seems like the most natural definition of orthogonality.

\begin{definition}[\cite{Mustothesis}]\label{def:OQLS}
    Two quantum Latin squares \(\left(\psi_{ij}\right)_{1\leq i,j\leq n}\) and \(\left(\phi_{ij}\right)_{1\leq i,j\leq n}\) are \emph{orthogonal} if \[\left\{\psi_{ij}\otimes\phi_{ij}\colon\,i,j\in\{1,\dots,n\}\right\}\] is an orthonormal basis of \(\CC^q\otimes\CC^q\).
\end{definition}

Since the definition in \cite{Goyeneche} does not agree with the properties of an absolutely maximally entangled (AME) state, Rajchel-Mieldzioć \cite[Definition~54]{Rajchel} modified it to the following. To avoid confusion with the previous definition, we call the squares \emph{entangled} instead of \emph{orthogonal}.

\begin{definition}[\cite{Rajchel}]\label{def:pairofEQLS}
    A \emph{pair of entangled quantum Latin squares of order \(n\)} is an \(n\times n\) matrix \(\left(\psi_{ij}\right)_{1\leq i,j\leq n}\) with entries in \(\CC^n\otimes\CC^n\) such that
\begin{enumerate}[(i)]
    \item all \(n^2\) entries form an orthonormal basis of \(\CC^n\otimes\CC^n\),
    \item all rows satisfy \(\tr_S\left(\sum_k\ketbra{\psi_{ik}}{\psi_{jk}}\right)=\delta_{ij}\id\) for \(S\in\{\{1\},\{2\}\}\), and
    \item all columns satisfy \(\tr_S\left(\sum_k\ketbra{\psi_{ki}}{\psi_{kj}}\right)=\delta_{ij}\id\) for \(S\in\{\{1\},\{2\}\}\),
\end{enumerate}
    where \(\tr_S\) denotes the partial trace over the \(S\) system and \(\ketbra{\psi}{\phi}\) is the outer product of \(\psi\) and \(\phi\).
\end{definition}

With the above definition, a pair of entangled quantum Latin squares of order \(n\) is equivalent to the existence of an AME\((4,n)\) state \cite[Lemma~55]{Rajchel}.
Since classical orthogonal Latin squares are particular instances of entangled quantum Latin squares, an AME\((4,n)\) state exists for all \(n\geq2\) except possibly if \(n=2\) or \(n=6\) \cite{BoseShrikhandeParker}. The fact that there do not exist classical orthogonal Latin squares of order six is due to Tarry \cite{Tarry}, who, by checking all cases, proved that there is no solution to \emph{Euler's thirty-six officers problem}. Later, non-computational proofs were given by Stinson \cite{noMOLS6Stinson} and Dougherty \cite{noMOLS6Dougherty} among others.

\begin{theorem}[\cite{Tarry}]\label{thm:2MOLS6}
    There does not exist a pair of orthogonal Latin squares of order six.
\end{theorem}

Still, it remained open whether AME\((4,n)\) states exist for \(n=2\) and \(n=6\). Higuchi and Sudbery \cite{noAME42} proved that an AME\((4,2)\) state does not exist. Recently, Rather, Burchardt, Bruzda, Rajchel-Mieldzioć, Lakshminarayan and \.Zyczkowski \cite{36entangledofficers} showed that there exists a pair of entangled quantum Latin squares of order six, solving the so-called \emph{Thirty-six entangled officers of Euler}. Later, other constructions of AME\((4,6)\) states were found by Rather \cite{AME46} and Rather, Ramadas, Kodiyalam and Lakshminarayan \cite{infiniteAME46}.

\begin{theorem}[\cite{36entangledofficers}]\label{thm:2MEQLS6}
    There exists a pair of entangled quantum Latin squares of order six.
\end{theorem}

With these results, the existence of pairs of entangled quantum Latin squares of order \(n\) is solved. However, it still remained open whether there exist two (non-entangled) orthogonal quantum Latin squares of order six. We answer this question in the negative.

\begin{theorem}\label{thm:no2MOQLS6}
    There does not exist a pair of orthogonal quantum Latin squares of order six.
\end{theorem}

\section{Preliminaries}

\subsection{Latin squares}

A \emph{Latin square} of order \(n\) is an \(n\times n\) matrix containing elements of a set of size \(n\) such that each row and each column contains every element exactly once.

Two Latin squares \(A=\left(a_{ij}\right)_{1\leq i,j\leq n}\) and \(B=\left(b_{ij}\right)_{1\leq i,j\leq n}\) are \emph{orthogonal} if all \(n^2\) tuples \((a_{ij},b_{ij})\) are different.
A set of \emph{\(t\) mutually orthogonal Latin squares of order \(n\)}, abbreviated \(t\) MOLS\((n)\), are \(t\) Latin squares of order \(n\) that are pairwise orthogonal.

The \emph{Latin square graph} of a given Latin square of order \(n\) is the graph whose vertices are the \(n^2\) entry coordinates of the Latin square, where two vertices are adjacent if they are in the same row or column or if the entries are equal.

There are several notions of equivalence for Latin squares, see also \cite{numberofLS}:
\begin{itemize}
    \item Two Latin squares are \emph{isotopic} if they can be obtained from one another by permuting the rows, columns and symbols.
    %\item Two Latin squares are \emph{trisotopic} if they can be obtained from one another by permuting the rows, columns and symbols and by transposing.
    \item Two Latin squares are \emph{paratopic} if they can be obtained from one another by permuting the rows, columns and symbols, and permuting their roles.
\end{itemize}
For example, swapping the roles of rows and columns corresponds to transposing. Equivalently, two Latin squares are paratopic if they have isomorphic Latin square graphs (all entries in the same row or column or with the same symbol correspond to maximum cliques).

The definitions of isotopy and paratopy can be extended to mutually orthogonal Latin squares. In this context, we are allowed to permutate the symbols independently within different squares, but the rows and columns should be permuted among all Latin squares simultaneously. See also \cite{enumerationMOLS}.

\subsection{Equivalence of quantum Latin squares}

Two quantum Latin squares are \emph{isotopic} if they can be obtained from one another by
\begin{enumerate}[(i)]
    \item multiplying entries with a phase factor,
    \item permuting rows and columns, and
    \item\label{item:u} applying a unitary transformation on all entries.
\end{enumerate}
A quantum Latin square is \emph{classical} or \emph{not genuinely quantum} if, up to isotopy, all entries are contained in \(\left\{\ket{1},\dots,\ket{n}\right\}\), a fixed orthonormal basis.
This property is independent of the chosen basis by item~(\ref{item:u}) in the definition of isotopy.

Note that the notion of paratopy does not carry over to non-classical quantum Latin squares, since, for example, the roles of rows and symbols cannot be swapped whenever the quantum Latin square has more different entries than rows.

\subsection{Mutually orthogonal quantum Latin squares (MOQLS)}

Mutually orthogonal quantum Latin squares were introduced by Musto \cite{Mustothesis}. These so-called MOQLS have been investigated in \cite{nonclassical,GMOQLS} among others. 
\begin{definition}[MOQLS \cite{Mustothesis}]
    \emph{\(t\) mutually orthogonal quantum Latin squares of order \(n\)} are \(t\) quantum Latin squares of order \(n\) such that any two of them are orthogonal (as in Definition~\ref{def:OQLS}).
\end{definition}
Just as for quantum Latin squares, two sets of \(t\) MOQLS\((n)\) are \emph{isotopic} if they can be obtained from one another by
\begin{enumerate}[(i)]
    \item multiplying entries with a phase factor,
    \item permuting rows and columns, simultaneously among all squares, and
    \item\label{item:lu} for each square, applying a unitary transformation on all entries of each square.
\end{enumerate}
We say that a set of \(t\) MOQLS\((n)\) is \emph{classical} or \emph{not genuinely quantum} if, up to isotopy, all entries are contained in \(\left\{\ket{1},\dots,\ket{n}\right\}\). Again, by item~(\ref{item:lu}) above, this property does not depend on the chosen orthonormal basis.

\begin{lemma}[{\cite[Corollary~17]{MustoVicary2019}}]\label{lem:standard}
    Up to isotopy, \(t\) MOQLS\((n)\) have \(\begin{tabular}{|c|c|c|}
    \hline
    \(\ket{1}\) & \(\cdots\) & \(\ket{n}\)\\
    \hline
    \end{tabular}\) in their first row.
\end{lemma}

\section{Patterns of quantum Latin squares}

The \emph{pattern} of a matrix is the binary matrix that is obtained from it by replacing each nonzero entry by a one.
A \emph{unitary pattern} is a pattern of a unitary matrix, see also \cite{unitarypatterns}.

Similarly, the pattern of a vector is the vector obtained from it by replacing each nonzero entry by a one.
The \emph{support} of a vector is the set of coordinates where it is nonzero.
The \emph{weight} of a vector is the size of its support, that is, the number of ones in its pattern.

We extend the notion of patterns to quantum Latin squares. Note that the pattern depends on the chosen orthonormal basis \(\left\{\ket{1},\dots,\ket{n}\right\}\).

\begin{definition}
    The \emph{pattern} of a quantum Latin square is the matrix that is obtained from it by replacing each entry by its pattern.
\end{definition}

For example:

\begin{table}[H]
    \centering
\begin{tabular}{|c|c|c|c|}
    \hline
    \(\ket{1}\) & \(\ket{2}\) & \(\ket{3}\) & \(\ket{4}\)\\
    \hline
    \(\ket{2}\) & \(\ket{1}\) & \(\ket{4}\) & \(\ket{3}\)\\
    \hline
    \(\ket{3}\) & \(\ket{4}\) & \(\frac1{\sqrt{2}}\left(\ket{1}+\ket{2}\right)\) & \(\frac1{\sqrt{2}}\left(\ket{1}-\ket{2}\right)\)\\
    \hline
    \(\ket{4}\) & \(\ket{3}\) & \(\frac1{\sqrt{2}}\left(\ket{1}-\ket{2}\right)\) & \(\frac1{\sqrt{2}}\left(\ket{1}+\ket{2}\right)\)\\
    \hline
\end{tabular}
\quad
\(\xrightarrow{\text{pattern}}\)
\quad
\begin{tabular}{|c|c|c|c|}
    \hline
    1000 & 0100 & 0010 & 0001\\
    \hline
    0100 & 1000 & 0001 & 0010\\
    \hline
    0010 & 0001 & 1100 & 1100\\
    \hline
    0001 & 0010 & 1100 & 1100\\
    \hline
\end{tabular}
\end{table}

We can always apply the following arguments to patterns of MOQLS. We see them as some kind of ``quantum sudoku rules'':

\begin{mdframed}
\begin{enumerate}
    \item[\#1] (Standard form.) By Lemma~\ref{lem:standard}, we can always start in standard form where the first row of the patterns of each quantum Latin square is of the form
    \[\begin{tabular}{|c|c|c|c|}
    \hline
    10\dots0 & 010\dots0 & \(\cdots\) & 0\dots01\\
    \hline
    \end{tabular}\]
    \item[\#2] (Unitary patterns.) Each row and each column of a quantum Latin square form a unitary pattern, by definition of quantum Latin squares. In particular, if there is an entry of weight one with support \(\{k\}\), then all other entries in the same row and column have a zero in coordinate \(k\).
    \item[\#3] (Zero overlap.) Starting from the second row onwards, if a pattern has a 1 on some position, there is a 0 on the same position in the pattern of an orthogonal quantum Latin square. Indeed, since the quantum Latin squares are mutually orthogonal, \(\psi_{ij}\otimes\phi_{ij}\) and \(\psi_{1k}\otimes\phi_{1k}=\ket{k}\otimes\ket{k}\) are orthogonal, so \(\braket{k}{\psi_{ij}}\braket{k}{\phi_{ij}}=0\).
\end{enumerate}
\end{mdframed}

\begin{theorem}\label{thm:n-1MOQLSn}
    There are at most \(n-1\) MOQLS of order \(n\). In case of equality, they are classical (and equivalent to a projective plane of order \(n\)).
\end{theorem}
\begin{proof}
    The inequality was already proven in \cite[Theorem~18]{MustoVicary2019}, but we do it again using the above rules. Rule \#1 says that, without loss of generality, we may assume that the first row of all matrices contains the states \(\ket{1},\dots,\ket{n}\). As in the classical proof, we look at the \(t\) entries on position \((2,1)\). By rule \#2, their pattern has a zero on the first position. By rule \#3, their supports partition the set \(\{2,\dots,n\}\). Thus, there are at most \(n-1\) MOQLS of order \(n\). If equality holds, all squares have a weight one entry on position \((2,1)\). The choice of this position was arbitrary, so all entries have weight one, meaning that the quantum Latin squares are classical.
\end{proof}

Before we prove the nonexistence of 2 MOQLS(6), we show that there do not exist non-classical 2 MOQLS\((n)\) if \(n=4\) or \(n=5\), solving some of the open cases in \cite[Theorem~3.7]{nonclassical}.

\begin{theorem}\label{thm:2MOQLS4}
    2 MOQLS(4) are classical.
\end{theorem}
\begin{proof}
    Let \(\Psi=\left(\psi_{ij}\right)_{1\leq i,j\leq4}\) and \(\Phi=\left(\phi_{ij}\right)_{1\leq i,j\leq4}\) be 2 MOQLS(4). By rule \#1, we may assume that the first row of both squares consists of \(\ket{1},\dots,\ket{4}\). Moreover, for all \(i\in\{2,3,4\}\) and \(j\in\{1,\dots,4\}\), the entry \(\psi_{ij}\) is orthogonal to \(\psi_{1j}=\ket{j}\) because they are in the same column. The same holds for the entries of \(\Phi\). Therefore, the patterns look like this (* can be either 0 or 1):
\[\Psi=
\begin{tabular}{|c|c|c|c|}
    \hline
    1000 & 0100 & 0010 & 0001\\
    \hline
    0*** & *0** & **0* & ***0\\
    \hline
    0*** & *0** & **0* & ***0\\
    \hline
    0*** & *0** & **0* & ***0\\
    \hline
\end{tabular}
\qquad
\text{and}
\quad
\Phi=
\begin{tabular}{|c|c|c|c|}
    \hline
    1000 & 0100 & 0010 & 0001\\
    \hline
    0*** & *0** & **0* & ***0\\
    \hline
    0*** & *0** & **0* & ***0\\
    \hline
    0*** & *0** & **0* & ***0\\
    \hline
\end{tabular}\]
    Suppose by contradiction that the squares are not classical. In other words, there is at least one entry that has weight two or more. Up to permutations, we can assume that it is the entry \(\psi_{21}\) in second row and the first column of \(\Psi\) and that its pattern is of the form 011*. Using rule \#3, we have \(\phi_{21}=\ket{4}\) and consequently \(\braket{4}{\psi_{21}}=0\), so the squares become:
\[\Psi=
\begin{tabular}{|c|c|c|c|}
    \hline
    1000 & 0100 & 0010 & 0001\\
    \hline
    0110 & *0** & **0* & ***0\\
    \hline
    0*** & *0** & **0* & ***0\\
    \hline
    0*** & *0** & **0* & ***0\\
    \hline
\end{tabular}
\qquad
\text{and}
\quad
\Phi=
\begin{tabular}{|c|c|c|c|}
    \hline
    1000 & 0100 & 0010 & 0001\\
    \hline
    0001 & *0** & **0* & ***0\\
    \hline
    0*** & *0** & **0* & ***0\\
    \hline
    0*** & *0** & **0* & ***0\\
    \hline
\end{tabular}\]
    The remaining states in the first column of \(\Psi\) cannot both be \(\ket{4}\) because they are orthogonal, so without loss of generality \(\psi_{31}\) has pattern 01**. Since it is orthogonal to \(\psi_{21}\), its pattern is 011*, but then \(\phi_{31}=\ket{4}\) by the same argument as before (rule \#3), a contradiction because \(\phi_{21}=\ket{4}=\phi_{31}\) are in the same column and therefore orthogonal.
\end{proof}

\begin{theorem}\label{thm:2MOQLS5}
    2 MOQLS(5) are classical.
\end{theorem}
\begin{proof}
    Consider 2 MOQLS(5). We apply rule \#1 and rule \#2 to get the squares in the following form:
\[\Psi=
\begin{tabular}{|c|c|c|c|c|}
    \hline
    10000 & 01000 & 00100 & 00010 & 00001\\
    \hline
    0**** & *0*** & **0** & ***0* & ****0\\
    \hline
    0**** & *0*** & **0** & ***0* & ****0\\
    \hline
    0**** & *0*** & **0** & ***0* & ****0\\
    \hline
    0**** & *0*** & **0** & ***0* & ****0\\
    \hline
\end{tabular}
\quad
\text{and }
\,
\Phi=
\begin{tabular}{|c|c|c|c|c|}
    \hline
    10000 & 01000 & 00100 & 00010 & 00001\\
    \hline
    0**** & *0*** & **0** & ***0* & ****0\\
    \hline
    0**** & *0*** & **0** & ***0* & ****0\\
    \hline
    0**** & *0*** & **0** & ***0* & ****0\\
    \hline
    0**** & *0*** & **0** & ***0* & ****0\\
    \hline
\end{tabular}\]
    Assume, by contradiction, that there is an entry of weight two or more.

    \textbf{Case 1: all entries have weight at most two.} Without loss of generality, let \(\psi_{21}\) have pattern 01100. The remaining entries in the second row cannot all be contained in a \(3\)-dimensional vector space (3-space for short), so there is another entry with the exact same pattern 011000. Let \(\psi_{24}\) be that other entry. Same for the first column; let \(\psi_{31}\) have this pattern as well:
\[\Psi=
\begin{tabular}{|c|c|c|c|c|}
    \hline
    10000 & 01000 & 00100 & 00010 & 00001\\
    \hline
    01100 & *0*** & **0** & 01100 & ****0\\
    \hline
    01100 & *0*** & **0** & ***0* & ****0\\
    \hline
    0**** & *0*** & **0** & ***0* & ****0\\
    \hline
    0**** & *0*** & **0** & ***0* & ****0\\
    \hline
\end{tabular}
\quad
\text{and }
\,
\Phi=
\begin{tabular}{|c|c|c|c|c|}
    \hline
    10000 & 01000 & 00100 & 00010 & 00001\\
    \hline
    0**** & *0*** & **0** & ***0* & ****0\\
    \hline
    0**** & *0*** & **0** & ***0* & ****0\\
    \hline
    0**** & *0*** & **0** & ***0* & ****0\\
    \hline
    0**** & *0*** & **0** & ***0* & ****0\\
    \hline
\end{tabular}\]
    After applying rule \#2 and rule \#3, we get:
\[\Psi=
\begin{tabular}{|c|c|c|c|c|}
    \hline
    10000 & 01000 & 00100 & 00010 & 00001\\
    \hline
    01100 & *00** & *00** & 01100 & *00*0\\
    \hline
    01100 & *00** & *00** & ***0* & ****0\\
    \hline
    000** & *0*** & **0** & ***0* & ****0\\
    \hline
    000** & *0*** & **0** & ***0* & ****0\\
    \hline
\end{tabular}
\quad
\text{and }
\,
\Phi=
\begin{tabular}{|c|c|c|c|c|}
    \hline
    10000 & 01000 & 00100 & 00010 & 00001\\
    \hline
    000** & *0*** & **0** & *000* & ****0\\
    \hline
    000** & *0*** & **0** & ***0* & ****0\\
    \hline
    0**** & *0*** & **0** & ***0* & ****0\\
    \hline
    0**** & *0*** & **0** & ***0* & ****0\\
    \hline
\end{tabular}\]
    We have \(\braket{3}{\psi_{i2}}=0\) for \(i\in\{1,2,3\}\), so \(\braket{3}{\psi_{42}}\neq0\) or \(\braket{3}{\psi_{52}}\neq0\). Up to permutation of the rows, we may assume that \(\braket{3}{\psi_{42}}\neq0\). Now look at the grey cells in the table:
\[\Psi=
\begin{tabular}{|c|c|c|c|c|}
    \hline
    10000 & 01000 & 00100 & 00010 & 00001\\
    \hline
    \cellcolor{black!30}01100 & *00** & *00** & \cellcolor{black!30}01100 & *00*0\\
    \hline
    \cellcolor{black!30}01100 & *00** & *00** & ***0* & ****0\\
    \hline
    000** & \cellcolor{black!30}*01** & **0** & ***0* & ****0\\
    \hline
    000** & *0*** & **0** & ***0* & ****0\\
    \hline
\end{tabular}
\quad
\text{and }
\,
\Phi=
\begin{tabular}{|c|c|c|c|c|}
    \hline
    10000 & 01000 & 00100 & 00010 & 00001\\
    \hline
    \cellcolor{black!30}000** & *0*** & **0** & \cellcolor{black!30}*000* & ****0\\
    \hline
    \cellcolor{black!30}000** & *0*** & **0** & ***0* & ****0\\
    \hline
    0**** & \cellcolor{black!30}*00** & **0** & ***0* & ****0\\
    \hline
    0**** & *0*** & **0** & ***0* & ****0\\
    \hline
\end{tabular}\]
    We prove that the four grey entries in \(\Phi\) are pairwise orthogonal. By definition, states in the same row or column are orthogonal. For pairs of entries in a different row and column, it suffices to prove that the corresponding entries in \(\Psi\) are \emph{not} orthogonal, since the tensor products \(\psi_{ij}\otimes\phi_{ij}\) are orthogonal (by definition of MOQLS). The state \(\psi_{42}\) is not orthogonal to \(\psi_{21}\), \(\psi_{31}\) and \(\psi_{24}\) because their support overlaps in exactly one position.
    The states \(\psi_{24}\) and \(\psi_{31}\) are not orthogonal because they are both orthogonal to \(\psi_{21}\) and have the same support of size two (so in fact, \(\psi_{24}=\psi_{31}\) up to a phase factor). Thus, the four grey entries in \(\Phi\) are pairwise orthogonal, yielding a contradiction because they are contained in the 3-space spanned by \(\ket{1}\), \(\ket{4}\) and \(\ket{5}\).
    
    \textbf{Case 2: there is an entry of weight at least three.} Without loss of generality, that entry is \(\psi_{21}\)  and its pattern is of the form 0111*. We will derive a contradiction by looking at the second row of the orthogonal squares. By rule \#3, the state \(\phi_{21}\) must have pattern 00001 and therefore \(\psi_{21}\) has pattern 01110. By rule \#2 in \(\Phi\), the other entries in the second row have a zero in the last coordinate. We get:
\[\Psi=
\begin{tabular}{|c|c|c|c|c|}
    \hline
    10000 & 01000 & 00100 & 00010 & 00001\\
    \hline
    01110 & *0*** & **0** & ***0* & ****0\\
    \hline
    \vdots & \vdots & \vdots & \vdots & \vdots\\
    \hline
\end{tabular}
\quad
\text{and }
\,
\Phi=
\begin{tabular}{|c|c|c|c|c|}
    \hline
    10000 & 01000 & 00100 & 00010 & 00001\\
    \hline
    00001 & *0**0 & **0*0 & ***00 & ****0\\
    \hline
    \vdots & \vdots & \vdots & \vdots & \vdots\\
    \hline
\end{tabular}\]
    Among the states \(\psi_{2i}\) for \(i\in\{2,3,4\}\), there is at least one of which the support intersects \(\{2,3,4\}\) because otherwise all three would be contained in the \(2\)-space spanned by \(\ket{1}\) and \(\ket{5}\), contradicting the fact that they are pairwise orthonormal. Thus, without loss of generality, the support of \(\psi_{22}\) intersects \(\{2,3,4\}\). It is orthogonal to \(\psi_{21}\), so it has pattern *011*. This implies \(\phi_{22}=\ket{1}\) (pattern 10000), and after applying rule \#2 in \(\Phi\), we get:
\[\Psi=
\begin{tabular}{|c|c|c|c|c|}
    \hline
    10000 & 01000 & 00100 & 00010 & 00001\\
    \hline
    01110 & *011* & **0** & ***0* & ****0\\
    \hline
    \vdots & \vdots & \vdots & \vdots & \vdots\\
    \hline
\end{tabular}
\quad
\text{and }
\,
\Phi=
\begin{tabular}{|c|c|c|c|c|}
    \hline
    10000 & 01000 & 00100 & 00010 & 00001\\
    \hline
    00001 & 10000 & 0*0*0 & 0**00 & 0***0\\
    \hline
    \vdots & \vdots & \vdots & \vdots & \vdots\\
    \hline
\end{tabular}\]
    The support of \(\psi_{23}\) and \(\psi_{24}\) can intersect \(\{2,3,4\}\) in at most one element by rule \#3. So, in order for those states to be orthogonal to \(\psi_{21}\), their support must be disjoint from \(\{2,3,4\}\). In other words, they span the \(2\)-space that is also spanned by \(\ket{1}\) and \(\ket{5}\). Hence the three states \(\psi_{21}\), \(\psi_{22}\) and \(\psi_{25}\) span the orthogonal space, that is, the \(3\)-space spanned by \(\ket{2}\), \(\ket{3}\) and \(\ket{4}\). In order for \[\begin{bmatrix}1&1&1\\0&1&1\\ *&*&*\end{bmatrix}\] to be a unitary pattern, all *'s must be 1 (this is easiest when looking at the columns). On the other hand, the weight of \(\psi_{25}\) cannot be three because of rule \#3, a contradiction.
\end{proof}

\section{MOQLS of order six}

In this section, we prove Theorem~\ref{thm:no2MOQLS6}.

\subsection{We may assume that one of 2 MOQLS(6) is classical}

\setlength{\tabcolsep}{4pt}

\begin{lemma}\label{lem:weight22}
    If \(\Psi\) is a quantum Latin square without entries of weight three or more, then there exists a classical Latin square \(\Psi'\) such that whenever \(\Psi\) and \(\Phi\) are orthogonal, \(\Psi'\) and \(\Phi\) are orthogonal as well.
\end{lemma}
\begin{proof}
    Let \(n\) be the order of \(\Psi\). We may assume that \(n\geq4\), otherwise the statement is trivial.
    We show that, as long as \(\Psi\) has entries of weight two, we can replace \(\Psi\) by a quantum Latin square with strictly less entries of weight two and no entries of weight three or more, while still being orthogonal to \(\Phi\). We can repeat this operation until all entries have weight one.
    
    Choose an entry \(\psi_{ij}\) of weight two, and assume without loss of generality that it has pattern 110\dots0. Let \(U\) be the unique unitary transformation that converts \(\psi_{ij}\) into \(\ket{1}\) and fixes \(\ket{3},\dots,\ket{n}\).
    Let \(\Psi'\) be the quantum Latin square obtained from \(\Psi\) by applying \(U\) to every entry with pattern 110\dots0, while keeping the other entries the same.
    First of all, the operation never increases the weight of any entry, and decreases the weight of \(\psi_{ij}\). Second, the operation preserves the orthogonality with other quantum Latin squares because whenever two entries in \(\Psi\) are orthogonal, the corresponding entries in \(\Psi'\) are orthogonal as well. The only case where this could go wrong is when one of two orthogonal entries has pattern 110\dots0, but then either the other entry has support \(\{1,2\}\) as well, or a support that is disjoint from it. 
\end{proof}

Note that the proof of Lemma~\ref{lem:weight22} no longer works when \(\Psi\) has weight three entries.

\begin{lemma}\label{lem:pattern3}
    If a \(6\times6\) unitary pattern has a row of weight three and no row of weight four or more, then either:
    \[(i) \begin{bmatrix}1&1&1&0&0&0\\1&1&1&0&0&0\\1&1&*&0&0&0\\&&\dots\end{bmatrix}\quad\text{or }\, (ii)\begin{bmatrix}1&1&1&0&0&0\\1&1&0&1&0&0\\1&0&1&1&0&0\\0&1&1&1&0&0\\&&\dots\end{bmatrix}\]
    \begin{enumerate}[(i)]
        \item there is another row of weight three with the same support and at least one more row of weight two or three whose support is contained in that of the given row of weight three, or
        \item there are four rows of weight three whose support is contained in a fixed set of size four, pairwise overlapping in two positions.
    \end{enumerate}
\end{lemma}
\begin{proof}
    Suppose without loss of generality that the first row is 111000. There must be at least two other rows whose support overlaps with that of the first row, because otherwise there would be four states in the same 3-space spanned by \(\ket{4}\), \(\ket{5}\) and \(\ket{6}\). Whenever the supports of two orthogonal states overlap, they must overlap in at least two elements. Therefore, if there is another row with the same support, we are in case (i). If not, there is a row with support \(\{1,2,4\}\) (up to permutations of \(\{1,\dots,6\}\) that fix \(\{1,2,3\}\)):
    \[\begin{bmatrix}1&1&1&0&0&0\\1&1&0&1&0&0\\&&\dots\end{bmatrix}\]
    Looking at the first and third column, there must be a row with ones in those columns in order for the supports of the columns to overlap in two or more positions. The support of this row must overlap in two or more positions with \(\{1,2,4\}\), so it is 101100. Similarly, there is a row with support \(\{2,3,4\}\) and we are in case (ii).
\end{proof}

\begin{lemma}\label{lem:pattern4}
    If a \(6\times6\) unitary pattern has a row of weight four and no row of weight five or more, then:
    \[(i) \begin{bmatrix}1&1&1&1&0&0\\1&1&1&1&0&0\\1&1&0&0&0&0\\&&\dots\end{bmatrix}\quad\text{or }\, (ii)\left[\begin{array}{cccccc}1&1&1&1&0&0\\ *&*&*&*&*&*\\ *&*&*&*&*&*\\ *&*&*&*&*&*\\&&\dots\end{array}\right]\begin{array}{c}\\\leftarrow\textnormal{weight}\geq3\\\leftarrow\textnormal{weight}\geq3\\\leftarrow\textnormal{weight}\geq3\\\\\end{array}\]
    \begin{enumerate}[(i)]
        \item there is another row of weight four with the same support and a row of weight two whose support is contained in the support of the given row of weight four, or
        \item there are four rows of weight at least three.
    \end{enumerate}
    In both cases, the unitary pattern has at least four rows of weight two or more.
\end{lemma}
\begin{proof}
    We may assume that the first row is 111100.
    There are five more rows and at most two of them can be of the form 0000**, so at least three rows have a support that overlaps with that of 111100 in at least one element. If all three have weight three or more, we are in case (ii). So suppose that one of them has weight one or two. By orthogonality with the first row, its weight must be two. Without loss of generality, we have
    \[\begin{bmatrix}1&1&1&1&0&0\\1&1&0&0&0&0\\&&\dots\end{bmatrix}\]
    The first and third column must correspond to orthogonal columns in the unitary matrix, so there is another row of the form 1*1***. It is orthogonal to the second row, so it is of the form 111***. If the fourth coordinate is a one, we are in case (i). If not, then it looks like
    \[\begin{bmatrix}1&1&1&1&0&0\\1&1&0&0&0&0\\1&1&1&0&*&*\\&&\dots\end{bmatrix}\]
    and we can apply a similar argument with the first and fourth column to show the existence of a row of the form 11*1**. Again, if its third coordinate is one, we are in case (i). If not, we have
    \[\begin{bmatrix}1&1&1&1&0&0\\1&1&0&0&0&0\\1&1&1&0&*&*\\1&1&0&1&*&*\\&&\dots\end{bmatrix}\]
    Applying the same argument on columns three and four, we see that there must be a row of the form **11**. It cannot be 001100, so its weight is at least three and we are in case (ii). 
    
    If there would be three rows of weight one, then all other rows would have weight at most three.
\end{proof}

%We can say more when the unitary pattern comes from a column of one of two orthogonal quantum Latin squares:

\begin{lemma}\label{lem:column4}
    If one of 2 MOQLS(6) in standard form has an entry of weight four, then the other square only has entries of weight one in that corresponding column.
\end{lemma}
\begin{proof}
    We assume that the squares are in standard form, so up to symmetry, the first row of the pattern is 100000. We can also assume the second row to be 011110. By rule \#3, the second row in the orthogonal pattern is 000001. Suppose by contradiction that the latter has a row of the form 011**0. We may assume this to be the third row. Therefore, the second and third column are both of the form 001***. Since the matrix is unitary, the second and third columns are orthogonal, so they have a nonzero entry somewhere in another position. In other words, there is another row of the form *11***. Hence the orthogonal unitary patterns look like
    \[\begin{bmatrix}1&0&0&0&0&0\\0&1&1&1&1&0\\0&0&0&*&*&*\\0&0&0&*&*&*\\&&\dots\end{bmatrix}\text{ and }\begin{bmatrix}1&0&0&0&0&0\\0&0&0&0&0&1\\0&1&1&*&*&0\\0&1&1&*&*&0\\&&\vdots\end{bmatrix}.\]
    The third and fourth row in the left pattern cannot both be 000001, so one of them is of the form 00011*. The second and fourth column in the left pattern must be orthogonal, so there is another row of the form 01*1**. We get
    \[\begin{bmatrix}1&0&0&0&0&0\\0&1&1&1&1&0\\0&0&0&1&1&*\\0&0&0&*&*&*\\0&1&*&1&*&*\\&&\dots\end{bmatrix}\text{ and }\begin{bmatrix}1&0&0&0&0&0\\0&0&0&0&0&1\\0&1&1&0&0&0\\0&1&1&*&*&0\\0&0&*&0&*&0\\&&\vdots\end{bmatrix}.\]
    For the pattern on the right to be unitary, the fifth row must be 000010. Thus, on the left, we get a zero in position \((5,5)\). The second and fifth row in the left unitary matrix are different and nonzero in the last three positions. This implies that the third and fourth row in the unitary matrix are equal, up to a phase factor, because they are determined as the vector with zeroes on the first three positions that is perpendicular to both rows, a contradiction.
\end{proof}

\begin{lemma}\label{lem:column3}
    If one of 2 MOQLS(6) in standard form has an entry of weight three, then the other square only has entries of weight one or two in that corresponding column.
\end{lemma}
\begin{proof}
    We assume the squares to be in standard form, so up to symmetry, the first row of the pattern is 100000.
    If the column has an entry of weight four, then we are done by Lemma~\ref{lem:column4}.
    So suppose all weights in that column are at most three. According to Lemma~\ref{lem:pattern3}, the pattern of that column is either
    \[\begin{bmatrix}1&0&0&0&0&0\\0&1&1&1&0&0\\0&1&1&1&0&0\\0&1&1&*&0&0\\0&0&0&0&*&*\\0&0&0&0&*&*\end{bmatrix}\text{ or }\begin{bmatrix}1&0&0&0&0&0\\0&1&1&1&0&0\\0&1&1&0&1&0\\0&1&0&1&1&0\\0&0&1&1&1&0\\0&0&0&0&0&1\end{bmatrix}.\]
    Following rule \#3, the unitary pattern of the corresponding column in the orthogonal quantum Latin square is either
    \[\begin{bmatrix}1&0&0&0&0&0\\0&0&0&0&*&*\\0&0&0&0&*&*\\0&0&0&*&*&*\\0&*&*&*&*&*\\0&*&*&*&*&*\end{bmatrix}\text{ or }\begin{bmatrix}1&0&0&0&0&0\\0&0&0&0&*&*\\0&0&0&*&0&*\\0&0&*&0&0&*\\0&*&0&0&0&*\\0&*&*&*&*&0\end{bmatrix}.\]
    In the first case, the second and third row in the unitary matrix are a basis of the space that is spanned by \(\ket{5}\) and \(\ket{6}\), so the other rows have zeroes on the fifth and sixth position. Hence, the fourth row is 000100 and the last two rows have a pattern of the form 0**000. In the second case, the only way to complete this to a unitary pattern is as a permutation matrix. In both cases, all rows have weight at most two.
\end{proof}

\begin{lemma}\label{lem:notboth4}
    If one of 2 MOQLS(6) in standard form has an entry of weight four, then the other square has no entries of weight four. Moreover, the other square has at most three columns containing an entry of weight three.
\end{lemma}
\begin{proof}
    Suppose by contradiction that both squares have entries of weight four. By Lemma~\ref{lem:column3} and Lemma~\ref{lem:column4}, one of them, say \(\Psi\), has at most three columns with entries of weight at least three. Suppose without loss of generality that \(\psi_{21}\) has pattern 001111. The unitary pattern determined by the second row of \(\Psi\) has at most three entries of weight at least three, so it belongs to case (i) of Lemma~\ref{lem:pattern4}. Hence, \(\psi_{22}\) has pattern 001111 as well, and we may assume that \(\psi_{23}\) has pattern 000011. Somewhere in the third column, there must be another entry with pattern of the form ****11, say \(\psi_{33}\). We distinguish two cases.
    
    \textbf{Case 1: \(\psi_{33}\) has weight two.} In other words, it has pattern 000011. After applying rule \#2 and rule \#3, we get:
\[
\Psi=
\scalebox{.8}{\begin{tabular}{|c|c|c|c|c|c|}
    \hline
    100000 & 010000 & 001000 & 000100 & 000010 & 000001\\
    \hline
    \cellcolor{black!30}001111 & 001111 & 000011 & ***0** & ****0* & *****0\\
    \hline
    0***** & *0**** & \cellcolor{black!30}000011 & ***0** & ****0* & *****0\\
    \hline
    \vdots & \vdots & \vdots & \vdots & \vdots & \vdots\\
    \hline
\end{tabular}}
\quad
\text{and }
\,
\Phi=
\scalebox{.8}{\begin{tabular}{|c|c|c|c|c|c|}
    \hline
    100000 & 010000 & 001000 & 000100 & 000010 & 000001\\
    \hline
    \cellcolor{black!30}010000 & 100000 & 000100 & ***0** & ****0* & *****0\\
    \hline
    0***** & *0**** & \cellcolor{black!30}**0000 & ***0** & ****0* & *****0\\
    \hline
    \vdots & \vdots & \vdots & \vdots & \vdots & \vdots\\
    \hline
\end{tabular}}\]
    Look at the grey cells in \(\Psi\). The corresponding entries, \(\psi_{21}\) and \(\psi_{33}\), cannot be orthogonal, because otherwise \(\psi_{23}\) and \(\psi_{33}\) would be equal up to a phase factor, a contradiction. Hence \(\phi_{33}\) is orthogonal to \(\phi_{21}\) by rule \#3. Similarly, it is orthogonal to \(\phi_{22}\), a contradiction.
    
    \textbf{Case 2: \(\psi_{33}\) has weight three or four.} Let \(i\in\{1,\dots,6\}\) be such that the \(i\)th row of \(\Phi\) contains an entry of weight four. %\(\Psi\) has three columns with an entry of weight at least three, so Lemma~\ref{lem:column3} and Lemma~\ref{lem:column4} imply that the pattern determined by the \(i\)th row of \(\Phi\) belongs to case (i) of Lemma~\ref{lem:pattern4}. 
    The \(i\)th row of \(\Phi\) has weight one entries in the first two columns because of Lemma~\ref{lem:column4}. By Lemma~\ref{lem:pattern4}, it has an entry of weight at least two in the third column. The weight cannot be four by Lemma~\ref{lem:column4}, and it cannot be three by Lemma~\ref{lem:column3}, so the \(i\)th row of \(\Phi\) has an entry of weight two in the third column. In particular, the pattern determined by the \(i\)th row of \(\Phi\) belongs to case (i) of Lemma~\ref{lem:pattern4}. Moreover, somewhere in the third column of \(\Phi\), there is another entry whose support overlaps with that of the third entry of the \(i\)th column of \(\Phi\). It has weight at most two (again, using Lemma~\ref{lem:column3}) and must be orthogonal, so it has the same pattern. We conclude that \(\Phi\) satisfies the conditions of Case~1, and we get a contradiction.

    We conclude that if \(\Psi\) has an entry of weight four, then \(\Phi\) has no entries of weight four. Moreover, \(\Phi\) cannot have four columns with an entry of weight three, because otherwise \(\Psi\) is in the scenario of Case 1 by Lemma~\ref{lem:column3}, and get a contradiction.
\end{proof}

\begin{lemma}\label{lem:weight3}
    There do not exist 2 MOQLS(6) in standard form such that one of them has an entry of weight three, no entries of weight at least four, and at most three columns with an entry of weight three.
\end{lemma}
\begin{proof}
    Let \(\Psi\) be a quantum Latin square with this property, and let \(\Phi\) be the orthogonal quantum Latin square. Because of the assumption, every unitary pattern coming from a row of \(\Psi\) with an entry of weight three, belongs to case (i) of Lemma~\ref{lem:pattern3}. In particular, for every entry of weight three, there is an entry in the same row with the same pattern.
    
    \textbf{Case 1: there is a column with a unitary pattern of the form (i) in Lemma~\ref{lem:pattern3}.} We may assume that:
\[
\Psi=
\scalebox{.8}{\begin{tabular}{|c|c|c|c|c|c|}
    \hline
    100000 & 010000 & 001000 & 000100 & 000010 & 000001\\
    \hline
    011100 & *0**** & **0*** & ***0** & ****0* & *****0\\
    \hline
    011100 & *0**** & **0*** & ***0** & ****0* & *****0\\
    \hline
    011*00 & *0**** & **0*** & ***0** & ****0* & *****0\\
    \hline
    0***** & *0**** & **0*** & ***0** & ****0* & *****0\\
    \hline
    0***** & *0**** & **0*** & ***0** & ****0* & *****0\\
    \hline
\end{tabular}}
\quad
\text{and }
\,
\Phi=
\scalebox{.8}{\begin{tabular}{|c|c|c|c|c|c|}
    \hline
    100000 & 010000 & 001000 & 000100 & 000010 & 000001\\
    \hline
    0000** & *0**** & **0*** & ***0** & ****0* & *****0\\
    \hline
    0000** & *0**** & **0*** & ***0** & ****0* & *****0\\
    \hline
    000*** & *0**** & **0*** & ***0** & ****0* & *****0\\
    \hline
    0***** & *0**** & **0*** & ***0** & ****0* & *****0\\
    \hline
    0***** & *0**** & **0*** & ***0** & ****0* & *****0\\
    \hline
\end{tabular}}\]
    Somewhere in the second row of \(\Psi\), there is another entry with pattern 011100, say it is \(\psi_{25}\). We use rule \#2 and rule \#3 to find:
\[
\Psi=
\scalebox{.8}{\begin{tabular}{|c|c|c|c|c|c|}
    \hline
    100000 & 010000 & 001000 & 000100 & 000010 & 000001\\
    \hline
    \cellcolor{black!30}011100 & *0**** & **0*** & ***0** & \cellcolor{black!30}011100 & *****0\\
    \hline
    \cellcolor{black!30}011100 & *0**** & **0*** & ***0** & ****0* & *****0\\
    \hline
    011000 & *00*** & *00*** & ***0** & ****0* & *****0\\
    \hline
    0000** & *0**** & **0*** & ***0** & ****0* & *****0\\
    \hline
    0000** & *0**** & **0*** & ***0** & ****0* & *****0\\
    \hline
\end{tabular}}
\quad
\text{and }
\,
\Phi=
\scalebox{.8}{\begin{tabular}{|c|c|c|c|c|c|}
    \hline
    100000 & 010000 & 001000 & 000100 & 000010 & 000001\\
    \hline
    \cellcolor{black!30}0000** & *0**** & **0*** & ***0** & \cellcolor{black!30}*0000* & *****0\\
    \hline
    \cellcolor{black!30}0000** & *0**** & **0*** & ***0** & ****0* & *****0\\
    \hline
    000100 & *0*0** & **00** & ***0** & ***00* & ***0*0\\
    \hline
    0**000 & *0**** & **0*** & ***0** & ****0* & *****0\\
    \hline
    0**000 & *0**** & **0*** & ***0** & ****0* & *****0\\
    \hline
\end{tabular}}\]
    If \(\psi_{25}\) were orthogonal to \(\psi_{31}\), it would be equal to \(\psi_{41}\), up to a phase factor, a contradiction. So they are not orthogonal, and \(\phi_{25}\) is orthogonal to both \(\phi_{21}\) and \(\phi_{31}\), so \(\phi_{25}\) has pattern 100000. We conclude that \(\psi_{35}\) cannot have pattern 011100, because then \(\phi_{35}\) would have pattern 100000 as well, for the same reason, giving a contradiction. Therefore, \(\psi_{36}\) has pattern 011100:
\[
\Psi=
\scalebox{.8}{\begin{tabular}{|c|c|c|c|c|c|}
    \hline
    100000 & 010000 & 001000 & 000100 & 000010 & 000001\\
    \hline
    011100 & *0**** & **0*** & ***0** & 011100 & *****0\\
    \hline
    011100 & *0**** & **0*** & ***0** & ****0* & 011100\\
    \hline
    011000 & *00*** & *00*** & ***0** & ****0* & *****0\\
    \hline
    0000** & *0**** & **0*** & ***0** & ****0* & *****0\\
    \hline
    0000** & *0**** & **0*** & ***0** & ****0* & *****0\\
    \hline
\end{tabular}}
\quad
\text{and }
\,
\Phi=
\scalebox{.8}{\begin{tabular}{|c|c|c|c|c|c|}
    \hline
    100000 & 010000 & 001000 & 000100 & 000010 & 000001\\
    \hline
    0000** & 00**** & 0*0*** & 0**0** & 100000 & 0****0\\
    \hline
    0000** & *0**** & **0*** & ***0** & ****0* & *****0\\
    \hline
    000100 & *0*0** & **00** & ***0** & ***00* & ***0*0\\
    \hline
    0**000 & *0**** & **0*** & ***0** & ****0* & *****0\\
    \hline
    0**000 & *0**** & **0*** & ***0** & ****0* & *****0\\
    \hline
\end{tabular}}\]
    We already have three columns with entries of weight three. In particular, the second column of \(\Psi\) can only have entries of weight one or two. We use this to derive a contradiction. First note that \(\psi_{22}\) cannot have pattern 001100. Whay not? Because it cannot be orthogonal to \(\psi_{31}\), since otherwise it would be equal to \(\psi_{41}\) up to a phase factor, which it is not: otherwise \(\phi_{21}\), \(\phi_{31}\) and \(\phi_{22}\) would be pairwise orthogonal and all have a pattern of the form 0000**. Similarly, \(\psi_{23}\), \(\psi_{32}\) and \(\psi_{33}\) all have patterns of the form *000**. There must be some entry in the second column with a one on the third coordinate, suppose it is \(\psi_{52}\). We claim that it has pattern 001100. Suppose by contradiction that \(\psi_{52}\) has pattern *010**. This is not orthogonal to any of the elements in grey in \(\Psi\), which implies that \(\phi_{52}\) is orthogonal to all the ones in grey in \(\Phi\):
\[
\Psi=
\scalebox{.8}{\begin{tabular}{|c|c|c|c|c|c|}
    \hline
    100000 & 010000 & 001000 & 000100 & 000010 & 000001\\
    \hline
    \cellcolor{black!30}011100 & *000** & *000** & ***0** & \cellcolor{black!30}011100 & *****0\\
    \hline
    \cellcolor{black!30}011100 & *000** & *000** & ***0** & ****0* & 011100\\
    \hline
    \cellcolor{black!30}011000 & *00*** & *00*** & ***0** & ****0* & *****0\\
    \hline
    0000** & *01*** & **0*** & ***0** & ****0* & *****0\\
    \hline
    0000** & *0**** & **0*** & ***0** & ****0* & *****0\\
    \hline
\end{tabular}}
\quad
\text{and }
\,
\Phi=
\scalebox{.8}{\begin{tabular}{|c|c|c|c|c|c|}
    \hline
    100000 & 010000 & 001000 & 000100 & 000010 & 000001\\
    \hline
    \cellcolor{black!30}0000** & 00**** & 0*0*** & 0**0** & \cellcolor{black!30}100000 & 0****0\\
    \hline
    \cellcolor{black!30}0000** & *0**** & **0*** & ***0** & ****0* & *****0\\
    \hline
    \cellcolor{black!30}000100 & *0*0** & **00** & ***0** & ***00* & ***0*0\\
    \hline
    0**000 & *00*** & **0*** & ***0** & ****0* & *****0\\
    \hline
    0**000 & *0**** & **0*** & ***0** & ****0* & *****0\\
    \hline
\end{tabular}}\]
    That is impossible, so the pattern of \(\psi_{52}\) is 001100. Recall that it cannot have weight three or more. Moreover, there must be another entry in the second column with a one on the third position, and following the same argument as before, the pattern of \(\psi_{62}\) is 001100 as well. Repeating this reasoning in the third column implies that \(\psi_{53}\) has pattern 010100, so it cannot be orthogonal to \(\psi_{52}\), a contradiction.

    \textbf{Case 2: all columns with an entry of weight three have a unitary pattern of the form (ii) in Lemma~\ref{lem:pattern3}.}

    We may assume that \(\Psi\) has the form:
\[
\Psi=
\scalebox{.8}{\begin{tabular}{|c|c|c|c|c|c|}
    \hline
    100000 & 010000 & 001000 & 000100 & 000010 & 000001\\
    \hline
    011100 & *0**** & **0*** & ***0** & ****0* & *****0\\
    \hline
    011010 & *0**** & **0*** & ***0** & ****0* & *****0\\
    \hline
    010110 & *0**** & **0*** & ***0** & ****0* & *****0\\
    \hline
    001110 & *0**** & **0*** & ***0** & ****0* & *****0\\
    \hline
    000001 & *0***0 & **0**0 & ***0*0 & ****00 & *****0\\
    \hline
\end{tabular}}
\]
    Every row pattern belongs to case (i) in Lemma~\ref{lem:pattern3}, so in the second row, there is another entry with pattern 011100.
    There are at least two more entries of weight three in the same column of that entry that are not in the last row. Without loss of generality, they are the entries in the third and fourth row. Because every row pattern belongs to case (i) in Lemma~\ref{lem:pattern3}, they have the same pattern as that row has in the first column. In particular, we are talking about the last column:
\[
\Psi=
\scalebox{.8}{\begin{tabular}{|c|c|c|c|c|c|}
    \hline
    100000 & 010000 & 001000 & 000100 & 000010 & 000001\\
    \hline
    011100 & *0**** & **0*** & ***0** & ****0* & 011100\\
    \hline
    011010 & *0**** & **0*** & ***0** & ****0* & 011010\\
    \hline
    010110 & *0**** & **0*** & ***0** & ****0* & 010110\\
    \hline
    001110 & *0**** & **0*** & ***0** & ****0* & *****0\\
    \hline
    000001 & *0***0 & **0**0 & ***0*0 & ****00 & *****0\\
    \hline
\end{tabular}}
\]
    Using rule \#3, we see that \(\Phi\) must have a pattern of the form
\[
\Phi=
\scalebox{.8}{\begin{tabular}{|c|c|c|}
    \hline
    100000 & \dots & 000001\\
    \hline
    0000** & \dots & *000*0\\
    \hline
    000*0* & \dots & *00*00\\
    \hline
    00*00* & \dots & *0*000\\
    \hline
    \dots & \dots & \dots\\
    \hline
\end{tabular}}
\]
    which gives a contradiction.
\end{proof}

\begin{theorem}\label{thm:oneof2MOQLS6isclassical}
    If there exist 2 MOQLS(6), then there exist 2 MOQLS(6) where one of them is classical.
\end{theorem}
\begin{proof}
    Recall that we may assume the squares to be in standard form according to rule \#1.
    They cannot have an entry of weight five or six by rule \#3.
    If one of the quantum Latin squares has an entry of weight four, then Lemma~\ref{lem:notboth4} says that the other one does not have entries of weight four, and it has at most three columns containing an entry of weight three.
    But also if both quantum Latin squares have no entries of weight at least four, then one of them has at most three columns containing an entry of weight three, because of Lemma~\ref{lem:column3}.
    Call this Latin square \(\Psi\).
    Lemma~\ref{lem:weight3} implies that \(\Psi\) does not have an entry of weight three or more. The statement follows from Lemma~\ref{lem:weight22}.
\end{proof}

\subsection{Checking twelve cases}

Now that we may assume one of the squares to be classical, we can translate the problem into graph theory terms.

An \emph{orthonormal representation} of a graph in \(\mathbb{R}^n\) (or \(\mathbb{C}^n\)) is a map from the vertex set of the graph to the unit vectors in \(\mathbb{R}^n\) (or \(\mathbb{C}^n\)) such that nonadjacent vertices are mapped to orthonormal vectors.

\begin{lemma}\label{lem:orthonormalrep}
    There exist 2 MOQLS(6) if and only if there is a Latin square graph on \(6\times6\) vertices whose complement has an orthonormal representation in \(\mathbb{C}^6\).
\end{lemma}
\begin{proof}
    By Theorem~\ref{thm:oneof2MOQLS6isclassical}, if there exist 2 MOQLS(6) \(\Psi\) and \(\Phi\), then one of them is classical, say \(\Psi\). By definition, \(\Phi\) is a quantum Latin square if and only if every two entries in the same row or column are orthonormal. Both squares are orthogonal if and only if any two entries of \(\Phi\) that correspond to the same entry in \(\Psi\), are orthonormal. Hence, the entries of \(\Phi\) are an orthonormal representation of the complement of the Latin square graph associated to \(\Psi\). Vice versa, if the complement of the Latin square graph of a given classical Latin square \(\Psi\) has an orthonormal representation in \(\mathbb{C}^6\), that representation defines a quantum Latin square that is orthogonal to it.
\end{proof}

Paratopic Latin squares have the same Latin square graph. As a corollary, in order to disprove the existence of 2 MOQLS(6), we only have to check pairs where one of the squares is one of the twelve classical Latin squares of order six up to paratopy (also called \emph{main classes} or \emph{species}). This classification was first found by Sch\"onhardt \cite{12LS6}, see also \cite{numberofLS} and \url{https://users.cecs.anu.edu.au/~bdm/data/latin_mc6.txt} or \url{https://users.monash.edu.au/~iwanless/data/MOLS/maxMOLS6.1.txt} for an online catalogue.

\begin{figure}[H]
    \centering
\begin{tabular}{cccc}
\begin{tabular}{|c|c|c|c|c|c|}
\hline
1&2&3&4&5&6\\
\hline
2&1&4&3&6&5\\
\hline
3&4&5&6&1&2\\
\hline
4&3&6&5&2&1\\
\hline
5&6&1&2&3&4\\
\hline
6&5&2&1&4&3\\
\hline
\end{tabular}
&
\begin{tabular}{|c|c|c|c|c|c|}
\hline
1&2&3&4&5&6\\
\hline
2&1&4&3&6&5\\
\hline
3&4&5&6&1&2\\
\hline
4&3&6&5&2&1\\
\hline
5&6&1&2&4&3\\
\hline
6&5&2&1&3&4\\
\hline
\end{tabular}
&
\begin{tabular}{|c|c|c|c|c|c|}
\hline
1&2&3&4&5&6\\
\hline
2&1&4&3&6&5\\
\hline
3&4&5&6&1&2\\
\hline
4&5&6&1&2&3\\
\hline
5&6&1&2&3&4\\
\hline
6&3&2&5&4&1\\
\hline
\end{tabular}
&
\begin{tabular}{|c|c|c|c|c|c|}
\hline
1&2&3&4&5&6\\
\hline
2&1&4&3&6&5\\
\hline
3&4&5&6&1&2\\
\hline
4&5&6&2&3&1\\
\hline
5&6&2&1&4&3\\
\hline
6&3&1&5&2&4\\
\hline
\end{tabular}
\\
\\
\begin{tabular}{|c|c|c|c|c|c|}
\hline
1&2&3&4&5&6\\
\hline
2&1&4&3&6&5\\
\hline
3&5&1&6&2&4\\
\hline
4&6&2&5&1&3\\
\hline
5&3&6&1&4&2\\
\hline
6&4&5&2&3&1\\
\hline
\end{tabular}
&
\begin{tabular}{|c|c|c|c|c|c|}
\hline
1&2&3&4&5&6\\
\hline
2&1&4&3&6&5\\
\hline
3&5&1&6&2&4\\
\hline
4&6&2&5&1&3\\
\hline
5&3&6&2&4&1\\
\hline
6&4&5&1&3&2\\
\hline
\end{tabular}
&
\begin{tabular}{|c|c|c|c|c|c|}
\hline
1&2&3&4&5&6\\
\hline
2&1&4&3&6&5\\
\hline
3&5&1&6&2&4\\
\hline
4&6&2&5&3&1\\
\hline
5&4&6&2&1&3\\
\hline
6&3&5&1&4&2\\
\hline
\end{tabular}
&
\begin{tabular}{|c|c|c|c|c|c|}
\hline
1&2&3&4&5&6\\
\hline
2&1&4&3&6&5\\
\hline
3&5&1&6&2&4\\
\hline
4&6&5&1&3&2\\
\hline
5&4&6&2&1&3\\
\hline
6&3&2&5&4&1\\
\hline
\end{tabular}
\\
\\
\begin{tabular}{|c|c|c|c|c|c|}
\hline
1&2&3&4&5&6\\
\hline
2&1&4&3&6&5\\
\hline
3&5&1&6&4&2\\
\hline
4&6&5&1&2&3\\
\hline
5&3&6&2&1&4\\
\hline
6&4&2&5&3&1\\
\hline
\end{tabular}
&
\begin{tabular}{|c|c|c|c|c|c|}
\hline
1&2&3&4&5&6\\
\hline
2&1&4&3&6&5\\
\hline
3&5&1&6&4&2\\
\hline
4&6&5&1&2&3\\
\hline
5&4&6&2&3&1\\
\hline
6&3&2&5&1&4\\
\hline
\end{tabular}
&
\begin{tabular}{|c|c|c|c|c|c|}
\hline
1&2&3&4&5&6\\
\hline
2&1&4&5&6&3\\
\hline
3&4&2&6&1&5\\
\hline
4&6&5&2&3&1\\
\hline
5&3&6&1&2&4\\
\hline
6&5&1&3&4&2\\
\hline
\end{tabular}
&
\begin{tabular}{|c|c|c|c|c|c|}
\hline
1&2&3&4&5&6\\
\hline
2&3&1&5&6&4\\
\hline
3&1&2&6&4&5\\
\hline
4&6&5&2&1&3\\
\hline
5&4&6&3&2&1\\
\hline
6&5&4&1&3&2\\
\hline
\end{tabular}
\end{tabular}
    \caption{There are twelve Latin squares of order six up to paratopy.}
    \label{fig:placeholder}
\end{figure}

\subsection{An algorithm to disprove the existence of an orthonormal representation}\label{sec:algorithm}

In ten out of twelve cases, we can solve the problem by reducing it to a discrete problem and using the computer.
Algorithm~\ref{alg:orthonormalrepresentation} outputs \texttt{False} if the graph does not have an orthonormal representation in \(\mathbb{R}^6\) (or \(\mathbb{C}^6\)). If it outputs \texttt{True}, we cannot conclude whether there exists an orthonormal representation. An explanation can be found below the pseudocode.

\begin{algorithm}[H]
\caption{{\sc complementCouldHaveOrthonormalRepresentation}(\(G\), \(L\))}\label{alg:orthonormalrepresentation}
\begin{algorithmic}[1]
\Require Graph \(G\), list \(L\) of triples of vertices of \(G\)
\Ensure \texttt{False} if the complement of \(G\) has no orthonormal representation in six dimensions such that all triples in \(L\) are mapped to three distinct points on a projective line;
\texttt{True} (inconclusive) otherwise\vspace{2mm}
\State addingEdges \(\gets \texttt{True}\)
\While{addingEdges}
\State addingEdges \(\gets \texttt{False}\)
\For{\(\ell\in L\)}
\For{\(x\in V(G)\)}
\If{\(x\) is adjacent to two out of the three vertices of \(\ell\)}
\State make \(x\) adjacent to all vertices of \(\ell\)
\State addingEdges \(\gets \texttt{True}\)
\EndIf
\EndFor
\EndFor
\EndWhile
\If{{\sc cliqueNumber}\((G)\geq7\)}
\State\Return\texttt{False}
\EndIf
\For{complete tripartite \(X\sqcup Y\sqcup Z\subseteq V(G)\) with \(|X|\leq3\), \(|Y|\leq3\), \(|Z|\leq3\) and \(|X|+|Y|+|Z|=7\)}
\State\Return {\sc couldBeDependent}\((G,L,X)\) or {\sc couldBeDependent}\((G,L,Y)\) or {\sc couldBeDependent}\((G,L,Z)\)
\EndFor
\end{algorithmic}
\end{algorithm}

In the first step, we add edges. Whenever a vertex is adjacent to two out of the three vertices of a triple in \(L\), it means that the corresponding projective point is orthogonal to two out of three (distinct) points on a projective line and hence orthogonal to all points of that line. In particular, it can be made adjacent to the third vertex.

In the second step, we calculate the clique number. If there is a clique of size seven, then there are seven independent vectors in \(\mathbb{C}^6\), a contradiction.

In the third step, we apply the following lemma:

\begin{lemma}\label{lem:tripartite}
    If \(X\), \(Y\) and \(Z\) are multisets of vectors in \(\mathbb{R}^n\) (or \(\mathbb{C}^n\)) such that every two vectors in different sets are orthogonal, then \(\dim\langle X\rangle+\dim\langle Y\rangle+\dim\langle Z\rangle=\dim\langle X,Y,Z\rangle\leq n\).
\end{lemma}
\begin{proof}
    This follows from Grassmann's formula since \(\pi\cap\pi^\perp=\{\mathbf{0}\}\) for any subspace \(\pi\).
\end{proof}

In graph theory terms, the condition translates to a complete tripartite subgraph (not necessarily induced) with tripartition classes \(X\), \(Y\) and \(Z\). We only consider the cases when those classes have at most three elements, to keep the code simple. If we would add an extra case where four points can be dependent (so a subgraph \(K_{3,4}\)) then the algorithm solves one more case, but it is not needed for the proof. If \(|X|+|Y|+|Z|=7\), then at least one of the sets \(X\), \(Y\) and \(Z\) is mapped to three dependent points, and the function of Algorithm~\ref{alg:dependentset} is called.

\begin{algorithm}[H]
\caption{{\sc couldBeDependent}(\(G\), \(L\), \(X\))}\label{alg:dependentset}
\begin{algorithmic}[1]
\Require Graph \(G\), list \(L\) of triples of vertices of \(G\), set \(X\) of vertices of \(G\) with \(|X|\leq3\)
\Ensure \texttt{False} if the complement of \(G\) has no orthonormal representation in six dimensions such that all triples in \(L\) are mapped to three distinct points on a projective line and \(X\) is mapped to a set of dependent vectors; \texttt{True} (inconclusive) otherwise\vspace{2mm}
\For{\(\{x,y\}\in \binom{X}{2}\)}
\Comment{Case 1: two vertices of \(X\) are mapped to the same point}
\If{\(x\not\sim y\) and \(\forall\ell\in L\colon\,\{x,y\}\not\subseteq\ell\)}
\State \(G' \gets G\) where \(x\) and \(y\) are merged
\If{{\sc complementCouldHaveOrthonormalRepresentation}\((G', L)\)}
\State\Return\texttt{True}
\EndIf
\EndIf
\EndFor
\If{\(|X|=3\) and \(|E(X)|\leq1\)}
\Comment{Case 2: \(X\) is mapped to three distinct points on a line}
\If{{\sc complementCouldHaveOrthonormalRepresentation}\((G, L\cup\{X\})\)}
\State\Return\texttt{True}
\EndIf
\EndIf
\State\Return\texttt{False}
\end{algorithmic}
\end{algorithm}

The algorithm runs for about one minute for each of the twelve Latin square graphs on \(6\times6\) vertices. Only two of them yield an inconclusive \texttt{True}. They both have a subsquare of order three.

\subsection{No subsquare of order three}

We are left with proving that if a Latin square of order six has a subsquare of order three, then it does not have a ``quantum orthogonal mate''.

\begin{lemma}\label{lem:subsquaredifferent}
If one of 2 MOQLS(6) is classical and has a subsquare of order three, then all nine entries in the corresponding block of the other square are different.
\end{lemma}
\begin{proof}
After suitable permutations, we can assume that the squares are/have pattern:
\[
\Psi=
\begin{tabular}{|c|c|c|c|c|c|}
    \hline
    \cellcolor{black!30}1&\cellcolor{black!30}2&\cellcolor{black!30}3&4&5&6\\
    \hline
    \cellcolor{black!30}2&\cellcolor{black!30}3&\cellcolor{black!30}1&&&\\
    \hline
    \cellcolor{black!30}3&\cellcolor{black!30}1&\cellcolor{black!30}2&&&\\
    \hline
    4&5&6&\cellcolor{black!30}&\cellcolor{black!30}&\cellcolor{black!30}\\
    \hline
    5&6&4&\cellcolor{black!30}&\cellcolor{black!30}&\cellcolor{black!30}\\
    \hline
    6&4&5&\cellcolor{black!30}&\cellcolor{black!30}&\cellcolor{black!30}\\
    \hline
\end{tabular}
\quad
\text{and }
\,
\Phi=
\begin{tabular}{|c|c|c|c|c|c|}
    \hline
    \cellcolor{black!30}100000 & \cellcolor{black!30}010000 & \cellcolor{black!30}001000 & 000100 & 000010 & 000001\\
    \hline
    \cellcolor{black!30}00**** & \cellcolor{black!30}*00*** & \cellcolor{black!30}0*0*** &&&\\
    \hline
    \cellcolor{black!30}0*0*** & \cellcolor{black!30}00**** & \cellcolor{black!30}*00*** &&&\\
    \hline
    0**0** & *0**0* & **0**0 &\cellcolor{black!30}&\cellcolor{black!30}&\cellcolor{black!30}\\
    \hline
    0***0* & *0***0 & **00** &\cellcolor{black!30}&\cellcolor{black!30}&\cellcolor{black!30}\\
    \hline
    0****0 & *0*0** & **0*0* &\cellcolor{black!30}&\cellcolor{black!30}&\cellcolor{black!30}\\
    \hline
\end{tabular}\]
Both squares are divided into four blocks such that the left square has entries 1, 2 and 3 in the grey blocks and entries 4, 5 and 6 in the white blocks. We focus on the lower left block in white, but by symmetry, our arguments are also true for the other blocks.

Suppose that two entries of \(\Phi\) in the same block are equal. The entries cannot be in the same row or column or correspond to the same number in \(\Psi\). Therefore, we may assume without loss of generality that \(\phi_{41}=\phi_{52}\). They have pattern 00*0*0.

\textbf{Case 1: \(\phi_{41}=\phi_{52}=\ket{3}\).} All entries in the fifth column of \(\Phi\) are orthogonal to \(\ket{3}\), except \(\phi_{65}\). As the only element of an orthonormal basis that is not orthogonal to \(\ket{3}\), the entry \(\phi_{65}\) must be equal to \(\ket{3}\). In the classical Latin square \(\Psi\), the number on position \((6,5)\) is either 1 or 2: it cannot be 3 since \(\psi_{13}\otimes\phi_{13}=\ket{3}\otimes\ket{3}=\ket{3}\otimes\phi_{65}\). If it is equal to 1, then all entries of \(\Phi\) corresponding to a 2 in \(\Psi\) are orthogonal to \(\ket3\) by being in a row or column with an entry equal to \(\ket3\). This is a contradiction because all entries corresponding to the same number in \(\Psi\) form an orthonormal basis. If \(\psi_{65}\) is equal to 2, then similarly all elements of \(\Phi\) corresponding to a 1 in \(\Psi\) are orthogonal to \(\ket3\), a contradiction.

\textbf{Case 2: \(\phi_{41}=\phi_{52}=\ket{5}\).} All entries in the fourth column of \(\Phi\) are orthogonal to \(\ket{5}\), except \(\phi_{64}\). Similarly, all entries in the sixth column of \(\Phi\) are orthogonal to \(\ket{5}\), except \(\phi_{66}\). Thus \(\phi_{64}=\phi_{66}=\ket{3}\), a contradiction.

\textbf{Case 3: \(\phi_{41}=\phi_{52}\) has weight two.} All entries in the fifth column of \(\Phi\) are orthogonal to \(\ket{5}\) and \(\phi_{41}=\phi_{52}\), except \(\ket{5}\) and \(\phi_{65}\). Thus, \(\phi_{65}=\ket{3}\).
All entries in the first column are orthogonal to \(\phi_{41}\) and either \(\ket{3}\) or \(\ket{5}\), except \(\phi_{41}\) itself and \(\phi_{21}\). So \(\phi_{21}\) has pattern 001010. Similarly, \(\phi_{32}\) has pattern 001010. But now every entry in the third column, different from \(\ket{3}\), has a pattern of the form **0*0*, a contradiction.
\end{proof}

\begin{lemma}\label{lem:subsquarescaled}
    If one of 2 MOQLS(6) in standard form is classical and has a subsquare of order three, then the other square is of the form:
\[
\Phi=
\begin{tabular}{|c|c|c|c|c|c|}
    \hline
    \cellcolor{black!30}\((1,0,0,0,0,0)\) & \cellcolor{black!30}\((0,1,0,0,0,0)\) & \cellcolor{black!30}\((0,0,1,0,0,0)\) & \hspace{5mm} & \hspace{5mm} & \hspace{5mm} \\
    \hline
    \cellcolor{black!30}& \cellcolor{black!30} & \cellcolor{black!30}&&&\\
    \hline
    \cellcolor{black!30} & \cellcolor{black!30} & \cellcolor{black!30}&&&\\
    \hline
    \(\lambda_{41}(0,a_2,a_3,0,a_5,a_6)\) & \(\lambda_{42}(b_1,0,b_3,b_4,0,b_6)\) & \(\lambda_{43}(c_1,c_2,0,c_4,c_5,0)\) &\cellcolor{black!30}&\cellcolor{black!30}&\cellcolor{black!30}\\
    \hline
    \(\lambda_{51}(0,-c_2,c_3,-c_4,0,c_6)\) & \(\lambda_{52}(a_1,0,-a_3,a_4,-a_5,0)\) & \(\lambda_{53}(-b_1,b_2,0,0,b_5,-b_6)\) &\cellcolor{black!30}&\cellcolor{black!30}&\cellcolor{black!30}\\
    \hline
    \(\lambda_{61}(0,b_2,b_3,b_4,b_5,0)\) & \(\lambda_{62}(c_1,0,c_3,0,c_5,c_6)\) & \(\lambda_{63}(a_1,a_2,0,a_4,0,a_6)\) &\cellcolor{black!30}&\cellcolor{black!30}&\cellcolor{black!30}\\
    \hline
\end{tabular}\]
\end{lemma}
\begin{proof}
With \(\Psi\) and \(\Phi\) as in Lemma~\ref{lem:subsquaredifferent}, let \(X=\{\phi_{41}, \phi_{52}, \phi_{63}\}\), \(Y=\{ \phi_{42},\phi_{53}, \phi_{61}\}\) and \(Z=\{\phi_{43}, \phi_{51}, \phi_{62}\}\).
The vector $\phi_{41}$ is orthogonal to $\phi_{42}$, $\phi_{43}$, $\phi_{51}$, $\phi_{52}$ since $\Phi$ is a quantum Latin square, and it is also orthogonal to $\phi_{53}$ and $\phi_{62}$ since $\Psi$ and $\Phi$ are orthogonal.
Hence, \(\phi_{41}\) is orthogonal to all vectors in \(Y\) and \(Z\). Similarly, we get that every two vectors in different sets are orthogonal. We can apply Lemma~\ref{lem:tripartite} to conclude that \(\dim\langle X\rangle+\dim\langle Y\rangle+\dim\langle Z\rangle\leq 6\). By Lemma~\ref{lem:subsquaredifferent}, the three vectors of each of the sets are distinct, so \(\dim\langle X\rangle=\dim\langle Y\rangle=\dim\langle Z\rangle=2\). We can now scale the entries and assume the form in the statement.
\end{proof}

%The vectors in the form of Lemma~\ref{lem:subsquarescaled} are no longer normalised, but vectors in the same row are still orthogonal.

\begin{lemma}\label{lem:subsquareweight2}
    If a quantum Latin square is of the form as in Lemma~\ref{lem:subsquarescaled}, then all entries in the white block below left have weight at most two.
\end{lemma}
\begin{proof}
    Let \(S=\left(s_{ij}\right)_{1\leq i,j\leq3}\) be the white block below left in \(\Phi\), up to scaling:
\[
S=
\begin{tabular}{|c|c|c|}
    \hline
    \((0,a_2,a_3,0,a_5,a_6)\) & \((b_1,0,b_3,b_4,0,b_6)\) & \((c_1,c_2,0,c_4,c_5,0)\)\\
    \hline
    \((0,-c_2,c_3,-c_4,0,c_6)\) & \((a_1,0,-a_3,a_4,-a_5,0)\) & \((-b_1,b_2,0,0,b_5,-b_6)\)\\
    \hline
    \((0,b_2,b_3,b_4,b_5,0)\) & \((c_1,0,c_3,0,c_5,c_6)\) & \((a_1,a_2,0,a_4,0,a_6)\)\\
    \hline
\end{tabular}\]
    The rows and columns of \(S\) still have the property that they are orthogonal, but not necessarily orthonormal.
    Let \(w(S)\) be the minimum weight among the nine entries of \(S\).

    Observe that there is a high amount of symmetry. Not only can we permute the first three columns and last three columns and permute the rows accordingly to get the square in the same form. We can also change the roles of columns and symbols, swapping the coordinates \(\{1,2,3\}\) and \(\{4,5,6\}\). %More specifically, the automorphism group is isomorphic to \(S_3\wr S_2\). \todo{check}
    
    \textbf{Case 1: \(w(S)=1\).} Up to symmetry between \(\{1,2,3\}\) and \(\{4,5,6\}\), and changing the role of columns and symbols if needed, we may assume that \(s_{11}=\ket{2}\), so \(a_2=1\) and \(a_3=a_5=a_6=0\).
\[
S=
\begin{tabular}{|c|c|c|}
    \hline
    \((0,1,0,0,0,0)\) & \((b_1,0,b_3,b_4,0,b_6)\) & \((c_1,c_2,0,c_4,c_5,0)\)\\
    \hline
    \((0,-c_2,c_3,-c_4,0,c_6)\) & \((a_1,0,0,a_4,0,0)\) & \((-b_1,b_2,0,0,b_5,-b_6)\)\\
    \hline
    \((0,b_2,b_3,b_4,b_5,0)\) & \((c_1,0,c_3,0,c_5,c_6)\) & \((a_1,1,0,a_4,0,0)\)\\
    \hline
\end{tabular}\]
    By orthogonality, we have \(b_2=c_2=0\).
\[
S=
\begin{tabular}{|c|c|c|}
    \hline
    \((0,1,0,0,0,0)\) & \((b_1,0,b_3,b_4,0,b_6)\) & \((c_1,0,0,c_4,c_5,0)\)\\
    \hline
    \((0,0,c_3,-c_4,0,c_6)\) & \((a_1,0,0,a_4,0,0)\) & \((-b_1,0,0,0,b_5,-b_6)\)\\
    \hline
    \((0,0,b_3,b_4,b_5,0)\) & \((c_1,0,c_3,0,c_5,c_6)\) & \((a_1,1,0,a_4,0,0)\)\\
    \hline
\end{tabular}\]

    \textbf{Case 1.1: \(a_1=0\).} Up to scaling, we have \(a_4=1\) and by orthogonality, \(b_4=c_4=0\).
\[
S=
\begin{tabular}{|c|c|c|}
    \hline
    \((0,1,0,0,0,0)\) & \((b_1,0,b_3,0,0,b_6)\) & \((c_1,0,0,0,c_5,0)\)\\
    \hline
    \((0,0,c_3,0,0,c_6)\) & \((0,0,0,1,0,0)\) & \((-b_1,0,0,0,b_5,-b_6)\)\\
    \hline
    \((0,0,b_3,0,b_5,0)\) & \((c_1,0,c_3,0,c_5,c_6)\) & \((0,1,0,1,0,0)\)\\
    \hline
\end{tabular}\]
    Again by orthogonality, \(b_1c_1=b_3c_3=b_5c_5=b_6c_6=0\).
    Suppose by contradiction that there is an entry of weight three. There are three possibilities:
    \begin{itemize}
        \item If \(b_1b_3b_6\neq0\), then \(s_{21}=\mathbf{0}\), a contradiction.
        \item If \(b_1b_5b_6\neq0\), then \(s_{13}=\mathbf{0}\), a contradiction.
        \item If at least three elements of \(\{c_1,c_3,c_5,c_6\}\) are nonzero, then \(s_{12}=\mathbf{0}\) or \(s_{23}=\mathbf{0}\) or \(s_{31}=\mathbf{0}\), a contradiction.
    \end{itemize}
    We conclude that if \(a_1=0\), then all entries have weight at most two.
    
    \textbf{Case 1.2: \(a_1\neq0\).} By orthogonality, \(b_1=c_1=0\).
\[
S=
\begin{tabular}{|c|c|c|}
    \hline
    \((0,1,0,0,0,0)\) & \((0,0,b_3,b_4,0,b_6)\) & \((0,0,0,c_4,c_5,0)\)\\
    \hline
    \((0,0,c_3,-c_4,0,c_6)\) & \((a_1,0,0,a_4,0,0)\) & \((0,0,0,0,b_5,-b_6)\)\\
    \hline
    \((0,0,b_3,b_4,b_5,0)\) & \((0,0,c_3,0,c_5,c_6)\) & \((a_1,1,0,a_4,0,0)\)\\
    \hline
\end{tabular}\]
    Suppose by contradiction that there is an entry of weight three. There are five possibilities:
    \begin{itemize}
        \item If \(b_3b_4b_6\neq0\), then \(s_{21}=\mathbf{0}\), a contradiction.
        \item If \(c_3c_4c_6\neq0\), then \(s_{12}=\mathbf{0}\), a contradiction.
        \item If \(b_3b_4b_5\neq0\), then \(s_{13}=\mathbf{0}\), a contradiction.
        \item If \(c_3c_5c_6\neq0\), then \(s_{23}=\mathbf{0}\), a contradiction.
        \item If \(a_4\neq0\), then \(b_4=c_4=0\). Looking at \(s_{13}\), we have \(c_5\neq0\). By orthogonality with \(s_{23}\), we get \(b_5=0\):
\[
S=
\begin{tabular}{|c|c|c|}
    \hline
    \((0,1,0,0,0,0)\) & \((0,0,b_3,0,0,b_6)\) & \((0,0,0,0,c_5,0)\)\\
    \hline
    \((0,0,c_3,0,0,c_6)\) & \((a_1,0,0,a_4,0,0)\) & \((0,0,0,0,0,-b_6)\)\\
    \hline
    \((0,0,b_3,0,0,0)\) & \((0,0,c_3,0,c_5,c_6)\) & \((a_1,1,0,a_4,0,0)\)\\
    \hline
\end{tabular}\]
        Now \(b_3\neq0\) and \(b_6\neq0\), implying \(s_{21}=\mathbf{0}\), a contradiction.
    \end{itemize}
    We conclude that also if \(a_1\neq0\), all entries have weight at most two.

\textbf{Case 2: \(w(S)=2\).} That is, there is an element of weight two and no element of weight one. By symmetry on $\{1,2,3\}$ and $\{4,5,6\}$, and changing the role of columns and symbols if necessary, we may assume either \(a_1=a_2=0\) or \(a_1=a_4=0\).

\textbf{Case 2.1: \(a_1=a_2=0\).} In this case, \(a_4\neq0\) and \(a_6\neq0\).
\[
S=
\begin{tabular}{|c|c|c|}
    \hline
    \((0,0,a_3,0,a_5,a_6)\) & \((b_1,0,b_3,b_4,0,b_6)\) & \((c_1,c_2,0,c_4,c_5,0)\)\\
    \hline
    \((0,-c_2,c_3,-c_4,0,c_6)\) & \((0,0,-a_3,a_4,-a_5,0)\) & \((-b_1,b_2,0,0,b_5,-b_6)\)\\
    \hline
    \((0,b_2,b_3,b_4,b_5,0)\) & \((c_1,0,c_3,0,c_5,c_6)\) & \((0,0,0,a_4,0,a_6)\)\\
    \hline
\end{tabular}\]
By orthogonality of the entries in the third row and column respectively, we have \(b_4=c_6=c_4=b_6=0\).
\[
S=
\begin{tabular}{|c|c|c|}
    \hline
    \((0,0,a_3,0,a_5,a_6)\) & \cellcolor{black!30}\((b_1,0,b_3,0,0,0)\) & \((c_1,c_2,0,0,c_5,0)\)\\
    \hline
    \cellcolor{black!30}\((0,-c_2,c_3,0,0,0)\) & \((0,0,-a_3,a_4,-a_5,0)\) & \((-b_1,b_2,0,0,b_5,0)\)\\
    \hline
    \((0,b_2,b_3,0,b_5,0)\) & \((c_1,0,c_3,0,c_5,0)\) & \((0,0,0,a_4,0,a_6)\)\\
    \hline
\end{tabular}\]
The entries in grey are orthogonal because they both correspond to a \(\ket{5}\) in \(\Psi\), so \(b_3c_3=0\),  which implies an element of weight one, a contradiction.

\textbf{Case 2.1: \(a_1=a_4=0\).} In this case, \(a_2\neq0\) and \(a_6\neq0\).
\[
S=
\begin{tabular}{|c|c|c|}
    \hline
    \((0,0,a_3,0,a_5,a_6)\) & \((b_1,0,b_3,b_4,0,b_6)\) & \((c_1,c_2,0,c_4,c_5,0)\)\\
    \hline
    \((0,-c_2,c_3,-c_4,0,c_6)\) & \((0,0,-a_3,a_4,-a_5,0)\) & \((-b_1,b_2,0,0,b_5,-b_6)\)\\
    \hline
    \((0,b_2,b_3,b_4,b_5,0)\) & \((c_1,0,c_3,0,c_5,c_6)\) & \((0,a_2,0,0,0,a_6)\)\\
    \hline
\end{tabular}\]
By orthogonality of the entries in the third row and column respectively, we have \(c_6=c_2=0\).
\[
S=
\begin{tabular}{|c|c|c|}
    \hline
    \((0,0,a_3,0,a_5,a_6)\) & \((b_1,0,b_3,b_4,0,b_6)\) & \((c_1,0,0,c_4,c_5,0)\)\\
    \hline
    \cellcolor{black!30}\((0,0,c_3,-c_4,0,0)\) & \((0,0,-a_3,a_4,-a_5,0)\) & \((-b_1,b_2,0,0,b_5,-b_6)\)\\
    \hline
    \((0,b_2,b_3,b_4,b_5,0)\) & \((c_1,0,c_3,0,c_5,0)\) & \((0,a_2,0,0,0,a_6)\)\\
    \hline
\end{tabular}\]
The entry in grey has weight at least two, so \(c_3\neq0\) and \(c_4\neq0\). It is orthogonal to \((0,0,a_3,0,a_5,a_6)\) and \((0,0,-a_3,a_4,-a_5,0)\), so, \(a_3=a_4=0\), creating an entry of weight one, a contradiction.

\textbf{Case 3: \(w(S)=3\).} There is an element of weight three and there are no elements of weight one or two. Up to symmetry between \(\{1,2,3\}\) and \(\{4,5,6\}\), and changing the role of columns and symbols if needed, we may assume that \(a_1=0\).
\[
S=
\begin{tabular}{|c|c|c|}
    \hline
    \((0,a_2,a_3,0,a_5,a_6)\) & \((b_1,0,b_3,b_4,0,b_6)\) & \((c_1,c_2,0,c_4,c_5,0)\)\\
    \hline
    \((0,-c_2,c_3,-c_4,0,c_6)\) & \((0,0,-a_3,a_4,-a_5,0)\) & \((-b_1,b_2,0,0,b_5,-b_6)\)\\
    \hline
    \((0,b_2,b_3,b_4,b_5,0)\) & \((c_1,0,c_3,0,c_5,c_6)\) & \((0,a_2,0,a_4,0,a_6)\)\\
    \hline
\end{tabular}\]
In this case, $a_2 \neq 0$, $a_3 \neq 0$, $a_4 \neq 0$, $a_5 \neq 0$ and $a_6 \neq 0$.

By orthogonality of the second row of \(B\), we have \(b_5=0\). Looking at the entry in the left below, we have \(b_3\neq0\). By orthogonality of the third row, \(c_3=0\) but also \(c_6=0\), implying an element of weight at most two, a contradiction.

\textbf{Case 4: \(w(S)=4\).} All entries have weight four, so all of the variables are nonzero.
    The orthogonality between vectors in the same row gives the nine equations
    \[\begin{matrix}
        a_3\overline{b_3}=-a_6\overline{b_6}&\qquad b_1\overline{c_1}=-b_4\overline{c_4}&\qquad c_2\overline{a_2}=-c_5\overline{a_5}\\
        c_3\overline{a_3}=-c_4\overline{a_4}&\qquad a_1\overline{b_1}=-a_5\overline{b_5}&\qquad b_2\overline{c_2}=-b_6\overline{c_6}\\
        b_3\overline{c_3}=-b_5\overline{c_5}&\qquad c_1\overline{a_1}=-c_6\overline{a_6}&\qquad a_2\overline{b_2}=-a_4\overline{b_4}
    \end{matrix}\]
    which gives, after multiplying them, \[|a_1a_2a_3b_1b_2b_3c_1c_2c_3|^2=-|a_4a_4a_5b_3b_4b_5c_3c_4c_5|^2\]
    contradicting the variables being nonzero.

Combining all four cases, we conclude that \(w(S)=1\) and that all elements in the subsquare of $\Phi$ have weight at most two.
\end{proof}

\begin{theorem}\label{thm:nosubsquare}
There do not exist 2 MOQLS(6) where one of them is classical and has a subsquare of order three.
\end{theorem}
\begin{proof}
    Suppose by contradiction that \(\Psi\) and \(\Phi\) are mutually orthogonal quantum Latin squares, where \(\Psi\) is classical and has a subsquare of order three. Put them in standard form:
\[
\Psi=
\begin{tabular}{|c|c|c|c|c|c|}
    \hline
    \cellcolor{black!30}1&\cellcolor{black!30}2&\cellcolor{black!30}3&4&5&6\\
    \hline
    \cellcolor{black!30}2&\cellcolor{black!30}3&\cellcolor{black!30}1&&&\\
    \hline
    \cellcolor{black!30}3&\cellcolor{black!30}1&\cellcolor{black!30}2&&&\\
    \hline
    4&5&6&\cellcolor{black!30}&\cellcolor{black!30}&\cellcolor{black!30}\\
    \hline
    5&6&4&\cellcolor{black!30}&\cellcolor{black!30}&\cellcolor{black!30}\\
    \hline
    6&4&5&\cellcolor{black!30}&\cellcolor{black!30}&\cellcolor{black!30}\\
    \hline
\end{tabular}
\quad
\text{and }
\,
\Phi=
\begin{tabular}{|c|c|c|c|c|c|}
    \hline
    \cellcolor{black!30}100000 & \cellcolor{black!30}010000 & \cellcolor{black!30}001000 & 000100 & 000010 & 000001\\
    \hline
    \cellcolor{black!30}00**** & \cellcolor{black!30}*00*** & \cellcolor{black!30}0*0*** &&&\\
    \hline
    \cellcolor{black!30}0*0*** & \cellcolor{black!30}00**** & \cellcolor{black!30}*00*** &&&\\
    \hline
    0**0** & *0**0* & **0**0 &\cellcolor{black!30}&\cellcolor{black!30}&\cellcolor{black!30}\\
    \hline
    0***0* & *0***0 & **00** &\cellcolor{black!30}&\cellcolor{black!30}&\cellcolor{black!30}\\
    \hline
    0****0 & *0*0** & **0*0* &\cellcolor{black!30}&\cellcolor{black!30}&\cellcolor{black!30}\\
    \hline
\end{tabular}\]
    Lemma~\ref{lem:subsquareweight2} implies that the entries in the white block below left of \(\Phi\) have weight at most two. By symmetry, the same is true for the grey block below right. Thus, in each column of \(\Phi\), there are at most two entries of weight three or four.
    
    \textbf{Case 1: \(\Phi\) has an entry of weight three or more.} By Lemma~\ref{lem:pattern3} and Lemma~\ref{lem:pattern4} there is another entry in the same column with the same support. In particular, such an entry has weight three and has pattern 000111 (in the grey block above on the left) or pattern 111000 (in the white block above on the right). Therefore, the unitary pattern formed by the second row of \(\Phi\) belongs to case (i) of Lemma~\ref{lem:pattern3}. Without loss of generality, we may assume that \(\Phi\) looks like this: %\(\phi_{21}\), \(\phi_{22}\), \(\phi_{31}\) and \(\phi_{32}\) all have pattern 000111.
\[
\Phi=
\begin{tabular}{|c|c|c|c|c|c|}
    \hline
    \cellcolor{black!30}100000 & \cellcolor{black!30}010000 & \cellcolor{black!30}001000 & 000100 & 000010 & 000001\\
    \hline
    \cellcolor{black!30}000111 & \cellcolor{black!30}000111 & \cellcolor{black!30}0*0*** &&&\\
    \hline
    \cellcolor{black!30}000111 & \cellcolor{black!30}000111 & \cellcolor{black!30}*00*** &&&\\
    \hline
    0**0** & *0**0* & **0**0 &\cellcolor{black!30}&\cellcolor{black!30}&\cellcolor{black!30}\\
    \hline
    0***0* & *0***0 & **00** &\cellcolor{black!30}&\cellcolor{black!30}&\cellcolor{black!30}\\
    \hline
    0****0 & *0*0** & **0*0* &\cellcolor{black!30}&\cellcolor{black!30}&\cellcolor{black!30}\\
    \hline
\end{tabular}\]
    The three white entries in the first column cannot all have pattern 0**000, so one of them, say \(\phi_{i1}\), where \(i\in\{4,5,6\}\), has a nonzero entry in position 4, 5 or 6. By Lemma~\ref{lem:subsquareweight2}, \(\phi_{i1}\) has weight at most two, and since it is orthogonal to an entry with pattern 000111, it has pattern 000011 or 000101 or 000110.
    Since \(\phi_{22}\), \(\phi_{21}\) and \(\phi_{31}\) are all contained in the 3-space spanned by \(\ket{4}\), \(\ket{5}\) and \(\ket{6}\) and \(\phi_{22}\) is orthogonal to \(\phi_{21}\) and \(\phi_{31}\) (\(\phi_{22}\) and \(\phi_{31}\) both correspond to \(3\) in \(\Psi\)), we have that \(\phi_{22}\) is equal to \(\phi_{i1}\) up to a phase factor. This is a contradiction, because \(\phi_{22}\) has weight three and \(\phi_{i1}\) has weight two.

    \textbf{Case 2: all entries of \(\Phi\) have weight at most two.} By Lemma~\ref{lem:weight22}, there exist 2 MOLS(6), contradicting Theorem~\ref{thm:2MOLS6}.
\end{proof}

We conclude that 2 MOQLS(6) do not exist.

\begin{proof}[Proof of Theorem~\ref{thm:no2MOQLS6}]
    Suppose by contradiction that there are 2 MOQLS(6). By Theorem~\ref{thm:oneof2MOQLS6isclassical}, we may assume that one of the two squares is classical. By Lemma~\ref{lem:orthonormalrep} and the algorithm in Section~\ref{sec:algorithm}, the other square has a subsquare of order three. But by Theorem~\ref{thm:nosubsquare} that is not possible either. We get a contradiction, so 2 MOQLS(6) do not exist.
\end{proof}

\section{Conclusion}

Using the notion of unitary patterns and orthonormal representations of graphs, we proved that there do not exist 2 MOQLS(6). We also proved that 2 MOQLS(4) and 2 MOQLS(5) are classical (Theorem~\ref{thm:2MOQLS4} and Theorem~\ref{thm:2MOQLS5}).

As a corollary, we improve upon a theorem by Han, Zang, Zhang and Tian \cite[Theorem~3.7]{nonclassical} saying that, if \(n\geq4\), there exist non-classical 2 MOQLS\((n)\), except possibly for \(n\in\{4,5,6,7\}\). We proved that the values 4, 5 and 6 are impossible, leaving only the case \(n=7\) as an open problem:

\begin{problem}
    Are 2 MOQLS(7) classical?
\end{problem}

Note that if one of 2 MOQLS is classical, the other one is not necessarily classical. An example of 2 MOQLS(9) where only one of the squares is not classical, is given by
\[
\begin{tabular}{|c|c|c|c|c|c|c|c|c|}
    \hline
    1&2&3&4&5&6&7&8&9\\
    \hline
    2&3&1&5&6&4&8&9&7\\
    \hline
    3&1&2&6&4&5&9&7&8\\
    \hline
    4&5&6&7&8&9&1&2&3\\
    \hline
    5&6&4&8&9&7&2&3&1\\
    \hline
    6&4&5&9&7&8&3&1&2\\
    \hline
    7&8&9&1&2&3&4&5&6\\
    \hline
    8&9&7&2&3&1&5&6&4\\
    \hline
    9&7&8&3&1&2&6&4&5\\
    \hline
\end{tabular}
\quad
\text{ and }
\quad
\begin{tabular}{|c|c|c|c|c|c|c|c|c|}
    \hline
    1&2&3&4&5&6&7&8&9\\
    \hline
    3&1&2&6&4&5&9&7&8\\
    \hline
    2&3&1&5&6&4&8&9&7\\
    \hline
    7&8&9&1&2&3&4&5&6\\
    \hline
    9&7&8&3&1&2&6&4&5\\
    \hline
    8&9&7&2&3&1&5&6&4\\
    \hline
    4&5&6&7&8&9&1&\(a\)&\(b\)\\
    \hline
    6&4&5&9&7&8&\(b\)&1&\(a\)\\
    \hline
    5&6&4&8&9&7&\(a\)&\(b\)&1\\
    \hline
\end{tabular}\]
where \(a=\frac1{\sqrt{2}}\left(\ket{2}+\ket{3}\right)\) and \(b=\frac1{\sqrt{2}}\left(\ket{2}-\ket{3}\right)\).
\newline

Lemma~\ref{lem:weight22} can be used to translate some results on MOLS to MOQLS. In a similar way as in the proof of Theorem~\ref{thm:n-1MOQLSn}, one can show that if there exist \(n-2\) MOQLS\((n)\), then there exist \(n-2\) MOLS\((n)\). Since \(n-2\) MOLS\((n)\) can always be extended to \(n-1\) MOLS\((n)\) \cite{stability}, the existence of \(n-2\) MOQLS\((n)\) implies the existence of \(n-1\) MOLS\((n)\), and hence a projective plane of order \(n\).

\paragraph{Acknowledgements.}
The authors acknowledge the support of the Spanish Ministry of Science, Innovation and Universities grant
PID2023-147202NB-I00.
Robin Simoens is supported by the Research Foundation Flanders (FWO) through the grant 11PG724N.
We thank Quentin Palazon, Tabriz Popatia and Albert Rico for helpful discussions.

\bibliographystyle{plain}
\bibliography{ref}

\vfill
\noindent\textsc{Simeon Ball}\\
\textsc{\small Department of Mathematics}\\[-1mm]
\textsc{\small Universitat Politècnica de Catalunya}\\[-1mm]
\textsc{\small C. Pau Gargallo 14, 08028 Barcelona, Spain}\\
{\it E-mail address:} {\href{mailto:simeon.michael.ball@upc.edu}{\url{simeon.michael.ball@upc.edu}}}\\

\noindent\textsc{Robin Simoens}\\
\textsc{\small Department of Mathematics: Analysis, Logic and Discrete Mathematics}\\[-1mm]
\textsc{\small Ghent University}\\[-1mm]
\textsc{\small Krijgslaan 297, 9000 Gent, Belgium}\\
\textsc{\small Department of Mathematics}\\[-1mm]
\textsc{\small Universitat Politècnica de Catalunya}\\[-1mm]
\textsc{\small C. Pau Gargallo 14, 08028 Barcelona, Spain}\\
{\it E-mail address:} {\href{mailto:Robin.Simoens@UGent.be}{\url{Robin.Simoens@UGent.be}}}\\

\end{document}